\documentclass{amsart}

\usepackage{cite}
\usepackage{booktabs}
\usepackage{amsmath,amssymb,amsfonts,amsthm}
\usepackage{algorithm}
\usepackage[noend]{algpseudocode}
\usepackage{chessboard}
\usepackage[T1]{fontenc}
\usepackage{commath}
\usepackage[caption=false]{subfig}
\usepackage{hyperref}
\usepackage{paralist}
\usepackage{abstract}
\usepackage{MnSymbol}

\MakeRobust{\Call}

\newcommand{\bigO}{\mathcal{O}}
\algrenewcommand\textproc{}%
\newtheorem{theorem}{Theorem}[section]
\newtheorem{lemma}[theorem]{Lemma}
\def\multiset#1{\ensuremath{\left(\kern-.2em\left(#1\right)\kern-.2em\right)}}

\begin{document}

\title{Counting Unreachable Single-Side Pawn Diagrams with Limitless Captures}
\author{Colin McDonagh \\ \href{mailto:cmcdonagh@cs.nuim.ie}{cmcdonagh@cs.nuim.ie}}
\maketitle

\begin{abstract}
Epiphainein counts unreachable single-side pawn diagrams (in chess) where pawns can move forward or diagonally-forward without limit whilst remaining on the board. Epiphainein is a serial calculation and takes a few seconds to calculate the number of unreachable diagrams on a regular $8 \times 8$ board. With a decent machine it should take roughly 4 hours to calculate the same on a $10 \times 10$ board.
\end{abstract}

\section{Introduction}
We count unreachable single-side pawn diagrams (in chess) found by attempting to match pawns to starting files, where a single side is either exclusively white or exclusively black, and a diagram \cite{diagrams} is the contents of a board's squares as opposed to a position which also accounts for side to move, castling rights and en-passant. Pawns can move forward or diagonally-forward without limit whilst remaining on the board, the latter of which can be seen as pawns capturing empty squares. The motivation for this project is Shirish Chinchalkar's ``An Upper Bound for the number of reachable positions'' \cite{unreachablePositions}. Finally the code is available at \cite{epiphainein}.

\section{Background}
\subsection{An Example}
Certain pawn diagrams are unreachable as pawns may only move forward or diagonally forward. Consider ``Fig.~\ref{unreachablePosition}''. In order for a pawn to be on \textit{a3}, it must have come from \textit{b2}. However, since there's a pawn on \textit{b2}, this diagram is unreachable.
\begin{figure}[H]
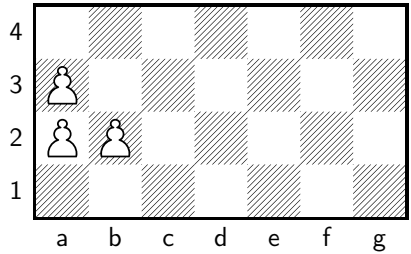

\begin{center}
\chessboard[setfen=8/P7/PP6/8 w - - 0 0,
            maxfield=g4,
            margintopwidth=0pt,
            showmover=false] 
\end{center}
\caption{The simplest unreachable diagram}
\label{unreachablePosition}
\end{figure}

We allow pawns to move diagonally forward without limit. In reality pawns can only move diagonally by capturing opposing chessmen, apart from the King.

\subsection{Naive Approach}
A pawn-square $s \in \mathcal{S}$ is a square which maps to files $\mu$ from which a pawn on $s$ could have started the game on, as shown in ``Fig.~\ref{originFiles}''. A set of pawn-squares $\nu \in 2^\mathcal{S}$ and the union of their potential starting files $M$ therefore form a bipartite graph $G = (\bigcup M, \nu, E)$.

\begin{figure}[H]
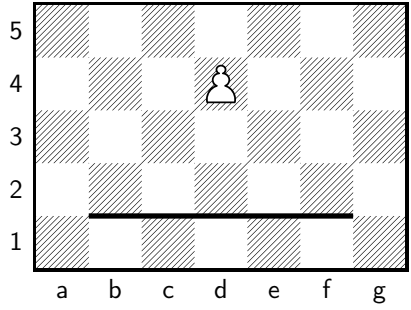

	\begin{center}
		\chessboard[setfen=8/3P8/8/8/8 w - - 0 0,
		maxfield=g5,
		pgfstyle=topborder,
		markregion={b1-f1},
		margintopwidth=0pt,
		showmover=false] 
	\end{center}
	\caption{$\mu_{d4} = [b, f]$}
	\label{originFiles}
\end{figure}

\begin{theorem} \label{diagramReachability}
A diagram is reachable iff its $G$ has a $\abs{\nu}$-perfect matching.
\end{theorem}

Every pawn-square in $\nu$ must be mapped to a starting square to be reachable, which is a $\abs{\nu}$-perfect matching in $G$. This matching can be found in $\bigO(\abs{\nu})$ via Horn's greedy \textit{EDF} scheduling algorithm \cite{horn} \footnote{For ordered input, i.e. pawns ordered by rank then file (row then column respectively)}. However, there are $\sum_{i=0}^{n}\binom{n (n-2)}{i}$ distinct diagrams on an $n \times n$ board, so an \textit{EDF} based serial computation for $n=8$ would take in the order of days or maybe weeks.

\section{A Bottom-Up Attempt at Counting Unreachable Diagrams}
\subsection{Overview}
For a range of starting files $u$, there is a maximally large set of pawn-squares $v$ whose potential starting files are a subset of $u$: $$\forall u \exists v \ni \forall s \in v, u_s \subseteq u \land \forall s \notin v, u_s \not \subset u$$ ``Fig.~\ref{nonEdgeInterval}'' shows an example of $v$ on a board.

\begin{figure}[H]
	\begin{center}
		\chessboard[maxfield=g5,
		pgfstyle=topborder,
		pgfstyle=color,
		color=black,
		colorbackfields={
			b2,
			c2, c3,
			d2, d3, d4,
			e2, e3,
			f2
		},
		margintopwidth=0pt,
		showmover=false] 
	\end{center}
	\caption{$v_{[b, f]}$}
	\label{nonEdgeInterval}
\end{figure}

Due to theorem \ref{diagramReachability}, $u \in \mathcal{U}$ produce unreachable diagrams in which there are $> \abs{u}$ pawns within $v_u$. More generally, every unreachable diagram is produced by at least some $U \in 2^\mathcal{U}$ and its equivalent pawn-squares $V_U = (v_u ; u \in U)$ where $\forall u \in U$, $v_u$ contains more than $\abs{u}$ pawns. Once we have the number of unreachable diagrams produced by each $U \in 2^\mathcal{U}$, we use the principle of inclusion-exclusion to determine the total number of unreachable diagrams.

\subsection{High-level Enumeration} \label{enumerationSection}
We order $u \in \mathcal{U}$ by $(u_0, u_1)$. We look to enumerate the satisfiable subset of $U \in 2^\mathcal{U}$ as efficiently as possible, where $U$ is satisfiable if it produces unreachable diagrams. The following explain how we achieve this efficiency.

\begin{theorem} \label{edgeTheorem}
	 Let ``edge'' be a quality of $U$ whose $(\bigcup U)_l = 1 = a$, and ``non\_edge'' a quality of $U$ whose $(\bigcup U)_l = 2 = b$ and $(\bigcup U)_r < n$. Then $U$ whose $(\bigcup U)_l > 2$ are either a a displacement of some $U \in (2^\mathcal{U})_{non\_edge}$ within $[\![2, n-1]\!]$ or a reversal of some $U \in (2^\mathcal{U})_{edge}$ displaced such that the rightmost file of the reversal is the  $n^{th}$ file. \footnote{It's also true that we only need compute one member of a file-reflected pair $(\exists, E)$, $\exists, E \in (2^\mathcal{U})_{non\_edge}$. However, due to an  implementation detail and because computing $U \in (2^\mathcal{U})_{edge}$ dominates $U \in (2^\mathcal{U})_{non\_edge}$ we haven't implemented it}
\end{theorem}

It's intuitive to just compute $U$ anchored on the $a$ file and produce $U \in 2^\mathcal{U}$ by considering images of anchored $U$. However it's necessary to have $edge$ and $non\_edge$ categories because they have different $\abs{V}$ as shown in ``Fig.~\ref{edgeVsNonedge}''.

\begin{figure}[H]
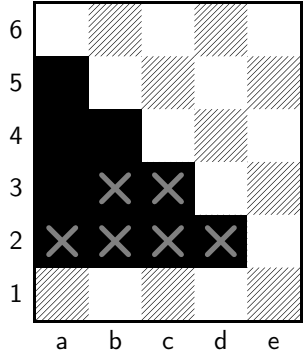

	\begin{center}
		\chessboard[maxfield=e6,
		pgfstyle=topborder,
		pgfstyle=color,
		color=black,
		colorbackfields={
			a2, a3, a4, a5,
			b2, b3, b4,
			c2, c3,
			d2
		},
		padding=-0.45em,
		color=gray,
		pgfstyle=cross,
		shortenstart=0.5ex,
		shortenend=0.5ex,
		markfields={
			a2,
			b2, b3,
			c2, c3,
			d2
		},
		margintopwidth=0pt,
		showmover=false] 
	\end{center}
	\caption{$V_1 = (v_{[a, d]})$ is filled with a black background. $V_2$ is equal to $(v_{[b, e]})$ left-shifted one square and is marked with gray $X$s. Although $U_1 = U_2 = ([a, d])$,  $\abs{\bigcup V_1} \neq \abs{\bigcup V_2}$}
	\label{edgeVsNonedge}
\end{figure}

\begin{theorem} \label{discontinuity}
	Let a continuous $U$ have the property that $\abs{\bigcup U} = (\bigcup U)_r - (\bigcup U)_l + 1$. Then non-continuous $U$ are multisets of continuous $U$.
\end{theorem}

Naturally, $U$ which are non-continuous can be split into continuous subsets which we only have to consider once, and then use to produce every $U \in 2^\mathcal{U}$.

\begin{theorem} \label{unsatCore}
	Let any unsatisfiable $U$ which we've had to enumerate be an ``unsat core''. Then a U which contains a file-reflection and/or displacement of a core is also unsatisfiable.
\end{theorem}

We check if $U$ contains a core by iterating over a copy of $U$; $M$. On each iteration check if $M$ displaced s.t$.$ $(\bigcup M)_l = 1$ is in an entry in our $cores$ dictionary; if so, $U$ is unsatisfiable, otherwise pop $M_1$ and continue until finding a core or $\abs{M} = 0$. This routine is preferable to checking satisfiability.

\begin{lemma} \label{noAlwaysUnsatu}
	Any $U$ which contains a $u$ for which $\abs{u} = \abs{v}$ or $\abs{u} = n$ is unsatisfiable.
\end{lemma}

Assume that references to $U$ hereon are to a $U$ which don't contain always unsat $u$ to avoid introducing unnecessary notation.

\begin{lemma}
	If $U$ is unsatisfiable then $\forall P, n < \sum_{v \in V}\abs{p_v}$ where $P = \{p_1, p_2, ..., p_{\abs{V}}\}$ are the pawns in $V$ and $p_i > \abs{U_i}$.
\end{lemma}

\begin{proof}
	$\forall u$, $u$ is potentially satisfiable as assumed from lemma \ref{noAlwaysUnsatu} onwards. If there are unlimited pawns available, then we can always produce unreachable diagrams by placing a pawn in every square $s \in \bigcup V$. Therefore any unsatisfiability must be caused by our limited pawn supply of $n$ pawns.
\end{proof}

\begin{theorem}
	If some $U \in (2^\mathcal{U})_{non\_edge}$ is unsatisfiable and $U^\prime$ is the ``non-edge equivalent'' of $U$, i.e$.$ $U$ right-shifted one file, then $U^\prime$ is also unsatisfiable.
\end{theorem}

$v \in V^\prime$ are either the same size as or smaller than their counterparts whereas $u \in U^\prime$ are identical to their counterparts s.t$.$ $U^\prime$ also requires $>n$ pawns. Therefore we first enumerate the entirety of $(2^\mathcal{U})_{edge}$ and feed the collected \textit{unsat cores} into a separate enumeration of $(2^\mathcal{U})_{non\_edge}$.

\begin{theorem} \label{lengthenedu}
	If $U + u$ is unsatisfiable and $\mu$ is $u$ right-lengthened and/or right-shifted then $U + \mu$ is also unsatisfiable.
\end{theorem}

Right-lengthening $u$ increases its pawn requirement above what was already unsatisfiably high. Right-shifting $u$ reduces already insufficient overlaps between $v$ and previous $\nu \in V$.

In light of the above, our enumeration is defined in ``Fig.~\ref{highLevelEnumAlg}''.
 
\begin{figure}[tp]
	\caption{High-level Enumeration}\label{highLevelEnumAlg}
	\begin{algorithmic}
		\State $unsat\_cores \gets$ \Call{enumerate}{$true$, $\emptyset$}
		\State \Call{enumerate}{$false$, $unsat\_cores$}
		\Function{enumerate}{$edge$, $unsat\_cores$}
			\State $unsat\_cores \gets unsat\_cores$ or $dict()$
			\State $u$ $\gets$ $\emptyset$
			\State $U \gets (u)$
			\State $sat \gets true$ \Comment{We consider $\emptyset$ initially pseudo-satisfiable}
			\State $max\_uLen \gets initMax\_uLen$ \Comment{$n-1$ ? $edge$ : $n-2$}
			\State \Comment{$max\_uLen$ tracks the longest $u$ which could potentially be added to $U$}
			\While{$U \nin END\_Us$} \Comment $\abs{END\_Us} = 2$
				\If{$sat$}
					\State \Call{append}{$U$, \Call{lexicographicallyNext\_u}{$U_{-1}$, $max\_uLen$}}
				\Else
					\State $toPop \gets$ \Call{num\_uToBacktrack}{$U$}
					\State $\mu \gets$ \Call{popNReturnLastPopped\_u}{$U$, $toPop$}
					\State \Comment{If one item is popped we return $U_{-1}$, if two then $U_{-2}$}
					\State $max\_uLen \gets (U_{-1_r} - U_{-1_l})$ \textbf{if} $toPop = 1$ \textbf{else} $initMax\_uLen$
					\State \Call{append}{U, \Call{lexicographicallyNext\_u}{$\mu$, $max\_uLen$}}
				\EndIf
			\If {any unsat core $\in U$} \Comment{Theorem \ref{unsatCore}}
			\State $sat \gets false$
			\State \textbf{continue}
			\EndIf
			\State $sat \gets$ \Call{countAndRecord}{$U$}
			\If {$!sat$}
			\State $unsat\_cores[U] = true$
			\State $unsat\_cores[$\Call{reverse}{$U$}$] = true$
			\EndIf
			\EndWhile
			\Return $unsat\_cores$
		\EndFunction
		\Function{lexicographicallyNextCandidate\_u}{$u$, $max\_uLen$}
			\If {$u_{new}$ = $(u_l, u_r + 1)$ isn't longer than $max\_uLen$ and $u_r + 1 < max\_ULen$} \Return $u_{new}$ \Comment{Theorem \ref{lengthenedu}. $maxULen = n$ ? $edge$ : $n-2$}
			\Else
			\State \Return $(u_l + 1, min(u_l + 1 + minNonEdge\_uLen, maxULen))$
			\State \Comment{$minNonEdge\_uLen =$ 3}
			\EndIf
		\EndFunction
		\Function{num\_uToBacktrack}{$U$}
			\If {the lexicographically last $u \in \mathcal{U}$ is already in $U$ or replacing $U_{-1}$ with $u$, $ U_{-1} <_{lex} u$, would result in unsatisfiability or discontinuity} \Return 2 
			\State \Comment{Theorem \ref{discontinuity}}
			\Else\ \Return 1
			\EndIf
		\EndFunction
	\end{algorithmic}
\end{figure}

\subsection{Diagrams contained within \textit{V}}
For a given $U$ we catalogue all unreachable diagrams contained within the corresponding $V$, i.e. where every pawn in the unreachable diagram is within $V$. To make this cataloging easier we first create a partition $\lambda$, $\lambda \vdash V$, with as few possible parts s.t$.$ $\forall \rho \in \lambda$ every square $s \in \rho$ is contained by the same subset of $U$. And we achieve this with a partition refinement strategy.

\subsubsection{Partition Refinement} A partition refinement \cite{refinement} incrementally partitions a family of sets $V$ into disjoint sets which collectively are $\lambda$. When adding $v$ to $\lambda$, we check $\forall \rho \in \lambda$ which $\abs{\rho \cap v} > 0$ and split those $\rho$ into $\rho_{v}$ and $\rho_{v^\prime}$. In order to remove $v$ from $\lambda$ it's normal to keep a union-find data structure \cite{unionFind}, however, we instead incrementally store $\lambda, \lambda \vdash M, M \in \set{(U_i ; 1 \leq i \leq j) ; 1 < j < \abs{U}}$ which collectively form $\Lambda$. This approach scales well because of some domain specific simplifications we can make to $\Lambda$. The following characterise our partitioning.

\begin{lemma}
	Given $\rho = I = v_a \bigcap v_b \bigcap ... \bigcap v_w$, $\rho$ can be simplified to $v_{[I_{maxl}, I_{minr}]}$, where $I_{maxl}$ is the rightmost left-file of any $u$ in $I$ and $I_{minr}$ the leftmost right-file of the same.
\end{lemma}

An intuition for this is, given that a number of similarly orientated congruent triangles are placed on top of a horizon intersect, the intersection is another smaller congruent triangle.

\begin{lemma} \label{onlyIntersection}
	 Given $\rho = I$, if $\abs{v \bigcap \rho} > 0$, then $\abs{v \bigcap \rho} = v_{[u_l, min(I_{minr}, u_r)]}$.
\end{lemma}

$I_{maxl} \leq u_l$ but $u_r$ can of course be to the left of $I_{minr}$.

\begin{lemma}
	Given $\rho = I \setminus D$, $D = v_\alpha \bigcup v_\beta ... \bigcup v_\omega$, if $\abs{v \bigcap \rho} > 0$, then $\abs{v \bigcap \rho} = v_{[u_l, min(I_{minr}, u_r)]} \setminus v_{[u_l, D_r]}$ where $D_r$ is the file of the rightmost square $s \in D$. 
\end{lemma}

$v_{[u_l, D_r]}$ can't be intersected by $v$ because whatever elements contributed the $D_r$ starting-square to $D$ have left-files to the left of $u_l$ and therefore cover $v$. Therefore the intersectable squares are those of $I$, given in lemma \ref{onlyIntersection}, less the excluded $v_{[u_l, D_r]}$ region.

\begin{theorem} \label{beforeu}
	Given $\rho \in \lambda$ are enumerated in the order they're added, $\abs{v \bigcap \rho} = 0$ and $I_{minr} < u_l$, then $\rho$ won't intersect any subsequent $\nu$.
\end{theorem}

\begin{proof}
	$\abs{v \bigcap \rho} = 0$ because $I_{minr} < u_l$. Given $\mu$ is the corresponding range of starting files of the subsequent $\nu$, then $u_l \leq \mu_l$, $I_{minr} < \mu_l$, and therefore $\abs{\nu \bigcap \rho} = 0$.
\end{proof}

As no subsequent $\nu$ can intersect $\rho$ we don't add it to $\lambda$.

\begin{theorem} \label{subsequentu}
	Given $\abs{v \bigcap \rho} > 0$, if $u_r <$ the file of the rightmost square in $\rho$, $\rho_r$, then $\rho_{v^\prime}$ is intersectable by a subsequent $\nu$. Conversely, if $\rho_r \leq u_r$, then $\rho_{v^\prime}$ isn't subsequently intersectable.
\end{theorem}

This is a purely geometric observation. See ``Fig.~\ref{subsequentlyNoIntersection}''. As a result we don't add any subsequently non-intersectable $\rho_{v^\prime}$ to $\lambda$.

\begin{figure}[H]
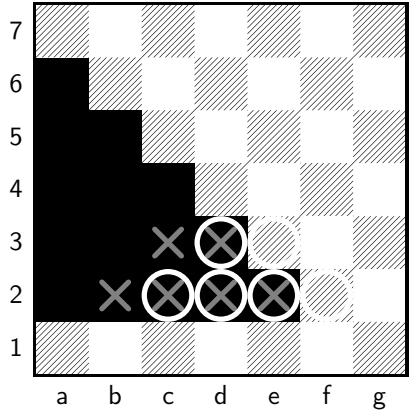

	\begin{center}
		\chessboard[maxfield=g7,
		pgfstyle=topborder,
		pgfstyle=color,
		color=black,
		colorbackfields={
			a2, a3, a4, a5, a6,
			b2, b3, b4, b5,
			c2, c3, c4,
			d2, d3,
			e2
		},
		padding=-0.45em,
		color=gray,
		pgfstyle=cross,
		shortenstart=0.5ex,
		shortenend=0.5ex,
		markfields={
			b2,
			c2, c3,
			d2, d3,
			e2
		},
		color=white,
		pgfstyle=circle,
		padding=-0.1em,
		markfields={
			c2,
			d2, d3,
			e2, e3,
			f2
		},
		shortenend=0.5ex,
		margintopwidth=0pt,
		showmover=false] 
	\end{center}
	\caption{$v_1 = v_{[a, e]}$ is filled with a black background. $\rho = v_1$. $V_1 = (v_1, v_2, v_3)$. $v_2 = v_{[b, e]}$ is marked with gray $X$s and $v_3 = v_{[c, f]}$ is marked by white circles. As $u_{2_r} = \rho_r =$ ``e'', the squares in $\rho$ that $v_3$ would have intersected had it been added before $v_2$, namely $c2$, $d2$, $d3$ and $e3$, belong to $\rho_v$. $\rho_{v^\prime}$ is not subsequently intersectable.}
	\label{subsequentlyNoIntersection}
\end{figure}

\subsubsection{Enumerating $\lambda$ Solutions} Having partitioned $V$, we enumerate every solution $P \in \mathcal{P_\lambda}$. Consider the following.

\begin{lemma}
	Let $v = V_{-1}$. Every solution to $V$ must also satisfy $V \setminus v$.
\end{lemma}

As a result we can produce every solution to $V$ from solutions to $V \setminus v$.

\begin{lemma}
	 Let $P$ be a solution to $V$ and $P^\prime$ be a solution to $V \setminus v$. Let $l \vdash V \setminus v$, and $l_v = \set{\rho; \rho \in l \land \abs{\rho \bigcap v} > 0}$. If $\forall \rho \in l_v, P^\prime_\rho = P_{\rho_v} + P_{\rho_{v^\prime}}$ then we say $P$ is produced from $P^\prime$.
\end{lemma}

We produce $P$ from $P^\prime$ by \begin{inparaenum} \item $\forall \rho \in l_v$, splitting pawns in $P^\prime_\rho$ into $P_{\rho_v}$ and $P_{\rho_{v^\prime}}$ \item $\forall \rho, \abs{\rho \bigcap v} = 0$ which remain subsequently intersectable, copying over $P^\prime_\rho$ to $P_\rho$ \item for an orphan partition $\alpha = v \setminus \bigcup (V \setminus v), \alpha \not = \{\emptyset\}$, placing up to whatever number of pawns haven't already been placed in $P$ in $\alpha$ \end{inparaenum} all s.t$.$ $V$ is satisfied. Our general approach to splitting $P^\prime_\rho$ and placing pawns in $\alpha$ is recursive.

\begin{theorem}
	Let $P$ be a solution for $\Lambda_V$ produced from $P^\prime$ which is a solution for $\Lambda_{V \setminus v}$. Let $C(X)$ be the number of unreachable diagrams which a solution $X$ represents. Let $l = \set{\rho; \rho \in \Lambda_{V \setminus v} \land \abs{\rho \bigcap v} > 0}$. Let $\alpha$ be an orphan part in $\Lambda_V$. Finally, let $N(\rho) = \binom{\abs{\rho_v}}{P_{\rho_v}} \binom{\abs{\rho_{v^\prime}}}{P_{\rho_{v^\prime}}}$ and $D(\rho) = \binom{\abs{\rho}}{P^\prime_\rho}$. Then $$C(P) = C(P^\prime) \binom{\abs{\alpha}}{P_\alpha} \prod_{\forall \rho \in l} \frac{N(\rho)}{D(\rho)}$$
\end{theorem}

We can count $C(P_{\Lambda_V})$ by continuously computing $C(P_{\Lambda_V})$ from $C(P_{\Lambda_{V \setminus v}})$, where the initial value $C(P_{\Lambda_{\emptyset}})$ is of course $0$.

\subsection{Counting All Unreachable Diagrams} We now consider how to generate all unreachable diagrams from the satisfiable subset of $U \in 2^\mathcal{U}_{edge} \bigcup 2^\mathcal{U}_{non\_edge}$.

\subsubsection{Disjoint Combinations of $U$} $\forall P_V$, extract the unique tuples $(e, w, \abs{p}, z, \abs{q})$ into $\mathcal{S}$ where $e \in \set{0, 1}$ indicates whether $U \in (2^\mathcal{U})_{edge}$, $w = (\bigcup U)_r - (\bigcup U)_l + 1$, $\abs{p} = \sum_{p \in P}p$, $z = \abs{U}$ and $\abs{q} = \abs{\bigcup_{v \in V} v}$. We enumerate $$\multiset{\mathcal{S}}_\Re = \set{S ; S \in \multiset{\mathcal{S}} \land \sum_{s \in S} s_w \leq n - E(S) \land \sum_{s \in S} s_{\abs{p}} \leq n \land \sum_{s \in S} s_e \leq 2}$$

where $\multiset{\mathcal{S}}$ is the infinite multiset of $S \in \mathcal{S}$ and $E(S)$ is the number of covered edge squares.\footnote{The number of covered edge squares isn't $\sum_{s \in S} s_e$ because one edge $U$ can cover both edge squares}.

Every unreachable diagram is counted via some $S \in \multiset{\mathcal{S}}_\Re$ if we consider two additional factors:

\begin{enumerate}
	\item Let $F(S)$ be the number of ways to uniquely and disjointly displace $s \in S_{non\_edge} = \set{s ; s \in S \land s_e = 0}$ s.t$.$ $s$ doesn't cover an edge file, and/or move to the opposing edge and reverse $s \in S_{edge}$.
	\item Let $R(S)$ be the number of ways to place between $0 \leq r \leq n-\abs{\sum_{s \in S} s_{\abs{p}}}$ remaining pawns in $n(n-2) - \sum_{s \in S} s_{\abs{q}}$ squares.
\end{enumerate}

Given $C(s)$ is the number of diagrams produced by $s$, then the number of diagrams produced by $S$ is $C(S) = \prod_{s \in S}C(s)F(S)R(S)$.

\subsubsection{Diagram Duplicity}
Diagrams aren't necessarily unique to some $S$. For e.g$.$ in ``Fig.~\ref{unreachablePositionPlusOne}''.

\begin{figure}[H]
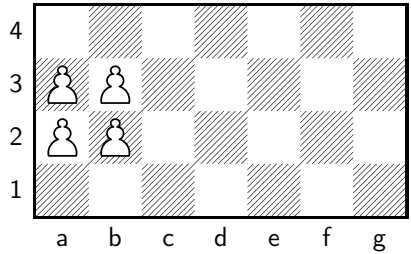

	\begin{center}
		\chessboard[setfen=8/PP6/PP6/8 w - - 0 0,
		maxfield=g4,
		margintopwidth=0pt,
		showmover=false
		] 
	\end{center}
	\caption{The above diagram may be produced by $\{\{[a, b]\}\}$, $\{\{[a, c]\}\}$ and $\{\{[a, b], [a, c]\}\}$ when we consider the factor $R$.}
	\label{unreachablePositionPlusOne}
\end{figure}

To count each diagram only once we use the inclusion-exclusion sieve:
\begin{theorem} The \# of unreachable diagrams is $\sum_{S \in \multiset{\mathcal{S}}_\Re}(-1)^{(\sum_{s \in S} s_z) + 1} C(S)$
\end{theorem}

\section{Sanity Check} We sanity check our method on boards of width $3 \leq n \leq 7$. For a given $n$, we enumerate all pawn diagrams in the $n(n-2)$ grid, and for each attempt to place pawns on starting squares using OR-Tools \cite{ortools}. We compute for $n = 7$ in about 2 hours because of a few observations: \begin{inparaenum} \item \label{i1} the horizontal reflection of an unsat diagram across the center of the board is unsat \item \label{i2} the row-displacement of an unsat diagram towards the starting files is unsat \item the horizontal displacement of an unsat diagram likely may also be unsat (see code at \cite{epiphainein}) \end{inparaenum}.

\section{Results}
The results obtained are shown in ``Table.~\ref{results}''.

\begin{table}[hbtp]
	\caption{Unreachable Diagrams} \label{results}
	\begin{center}
		\begin{tabular}{cccc}
			\toprule
			\textbf{n}& \textbf{\textit{$\#$Unreachable}}& \textbf{\textit{ \%Unreachable}}& \textbf{\textit{Approx$.$ Time}} \\
			\cmidrule(lr){1-1}\cmidrule(lr){2-2}\cmidrule(lr){3-3}\cmidrule(lr){4-4}
			3 & 0 & 0 & <1s \\
			4 & 18 & 11.04 & <1s \\
			5 & 550 & 11.12 & <1s \\
			6 & 16398 & 08.63 & <1s \\
			7 & 541782 & 06.20 & <1s \\
			8 & 20217623 & 04.35 & 3s \\
			9 & 851074312 & 03.02 & 3m \\
			10 & 40168190051 & 02.10 & 4h \\
			\bottomrule
		\end{tabular}
	\end{center}
\end{table}

{\raggedright

}
\end{document}